\documentclass[submission,copyright,creativecommons]{eptcs}
\providecommand{\event}{QPL 2016}
\usepackage{amsmath}
\usepackage{amsthm}
\usepackage{amssymb}
\usepackage{amsfonts}
\usepackage{bbm}
\usepackage{tikz}
\usepackage{braket}
\usepackage[numbers,compress,sort]{natbib}

\theoremstyle{plain}
\newtheorem{theorem}{Theorem}[section]
\newtheorem{proposition}[theorem]{Proposition}

\newtheorem{lemma}[theorem]{Lemma}

\theoremstyle{definition}
\newtheorem{procedure}[theorem]{Procedure}
\newtheorem{definition}[theorem]{Definition}
\newtheorem{remark}[theorem]{Remark}
\newtheorem{example}[theorem]{Example}
\newtheorem{problem}[theorem]{Problem}

\def\diag{\ensuremath{\mathrm{diag}}}
\def\End{\ensuremath{\mathrm{End}}}
\def\Rep{\ensuremath{\mathrm{Rep}}}
\def\Tr{\ensuremath{\mathrm{Tr}}}
\def\Re{\ensuremath{\mathrm{Re}}}
\def\T{\ensuremath{\mathrm{T}}}
\newcommand\cat[1]{\ensuremath{\mathbf{#1}}}
\newcounter{jamiecomment}
\newcommand\JVcomm[1]{\ensuremath{{}^{\color{red}\thejamiecomment}}\marginpar{\color{red}\tiny\raggedright \thejamiecomment: #1}\stepcounter{jamiecomment}}
\newcommand\ignore[1]{}

\begin{document}

\title{\bf Tight Reference Frame--Independent\\Quantum Teleportation}

\author{Dominic Verdon
\institute{Department of Computer Science \\ University of Oxford\\ Oxford, UK}
\email{dominic.verdon@cs.ox.ac.uk}
\and
Jamie Vicary
\institute{Department of Computer Science \\ University of Oxford\\ Oxford, UK}
\email{\quad jamie.vicary@cs.ox.ac.uk}
}
\def\titlerunning{Tight Reference Frame--Independent Quantum Teleportation}
\def\authorrunning{D.Verdon \& J. Vicary}

\ignore{\author{
\begin{tabular}{cc}
Dominic Verdon & Jamie Vicary
\\
Department of Computer Science & Department of Computer Science
\\
University of Oxford & University of Oxford
\\
\texttt{dominic.verdon@cs.ox.ac.uk}
&
\texttt{jamie.vicary@cs.ox.ac.uk}
\end{tabular}
}
}

\maketitle

\begin{abstract}
We give a tight scheme for teleporting a quantum state between two parties whose reference frames are misaligned by an action of a finite symmetry group. Unlike previously proposed schemes, ours requires no additional tokens or data to be passed between the participants; the same amount of classical information is transferred as for ordinary quantum teleportation, and the Hilbert space of the entangled resource is of the same size. In the terminology of Peres and Scudo, our protocol relies on classical communication of \emph{unspeakable} information.
\end{abstract}

\section{Introduction}

\subsection{The problem and our result}

Recently many authors have recognised the importance of developing a theory of quantum information which takes account of the reference frames by which a system's state is defined \cite{Bartlett2007,Marvian2013,Marvian2014,Kitaev2004}. It was recognised some time ago \cite{Enk2001} that a shared reference frame is a hidden implicit assumption in the conventional description of quantum teleportation, an assumption which it is reasonable to challenge, since it is clearly possible that two parties attempting teleportation across a large distance will be uncertain about their relative reference frames. Quantum teleportation is a foundational quantum protocol with important applications \cite{Gottesman1999,Gisin2002}; the problem of teleporting a quantum state between two parties who do not share a reference frame is therefore natural and important.

A successful reference frame--independent teleportation protocol is one where any observer with their own fixed reference frame will agree that Alice's quantum state has been successfully transferred to Bob, regardless of the orientation of Alice and Bob's frames; this was referred to in \cite{Chiribella2012} as `teleportation of unspeakable information'. Here we consider perfect teleportation, where Alice's state is transferred to Bob with certainty.

The effect of a change in reference frame alignment on the perceived state of a system may be encoded formally as the unitary action of a group of reference frame transformations on the Hilbert space of the system under consideration. It has been shown~\cite{Chiribella2012} that if the group of reference frame transformations contains $U(1)$ and the action is nontrivial, then reference frame--independent teleportation is impossible. In this paper we show that when the group of reference frame transformations is finite, reference frame--independent teleportation is possible in some nontrivial cases.

We will demonstrate that reference frame--independent (RFI) teleportation protocols correspond exactly to \emph{$G$-equivariant unitary error bases}, a structure we define, and develop methods of constructing these bases when they exist. In particular, we provide a simple  sufficient condition for the existence of RFI teleportation protocols for systems of dimension less than 5.

The nature of the classical channel through which the result of Alice's measurement is communicated is crucial to our new protocol. In the terminology of Peres and Scudo~\cite{Peres2002}, our protocol requires classical transmission of \emph{unspeakable information}, rather than of \emph{speakable information}. An example of unspeakable information is the choice of a direction in space, to be agreed by two parties who do not share a common directional reference; no amount of communication through a generic shared classical channel can decide the matter, but the fidelitous transfer from one party to the other of a single oriented physical system, such as an arrow, is sufficient.

We foresee two further applications of this result, both of which merit further investigation:

\begin{itemize}

\item \emph{Reference frame hiding.} Our result demonstrates that it is possible for two parties with secret reference frame configurations to perform teleportation of a quantum state without transmitting any information about those configurations, either to each other or to a third party. We have outlined such a scenario in the main example in Section 2.

\item \emph{Infinite groups of reference frame transformations.} Although reference frame--independent teleportation is impossible for an infinite Lie group, our result may be useful in the infinite case for developing protocols for imperfect quantum teleportation, a problem which has already been investigated \cite{Marzolino2015}. Firstly, it may be possible in some situations to render the group of reference frame transformations finite using limited prior communication or some other approach. Secondly, imperfect infinite-group reference frame--independent teleportation procedures may be attained as limits of perfect finite-group schemes.
\end{itemize}

The fundamental concept of a $G$-equivariant unitary error basis was developed from investigations in categorical quantum mechanics~\cite{Abramsky2009}. We characterised teleportation schemes as structures internal to the category of finite dimensional Hilbert spaces, and investigated corresponding structures in the category of unitary representations of a finite group; these corresponding structures were exactly the $G$-equivariant unitary error bases. Following their discovery, further investigation demonstrated their relevance to the problem of RFI teleportation. This exemplifies the utility of categorical quantum mechanics as a toolkit for developing new and interesting concepts in quantum information.

The outline of the paper is as follows. In Section~\ref{sec:example} we give an informal worked example of our procedure, and in Section~\ref{quantumteleportationsection} we provide a more formal presentation, along with further examples. In Section~\ref{sec:repg} we show how reference frame--independent teleportation schemes are related to ideas in categorical quantum mechanics. In Section~\ref{existenceconstruction} we prove a variety of existence, nonexistence and construction results for $G$-equivariant unitary error bases.

\subsection{Previous results}

It was demonstrated in \cite{Chiribella2012} that teleportation is impossible when the group of reference frame transformations is a infinite compact connected Lie group and the representation on the system to be teleported does not factor through a representation of a finite group. As observed in that paper, this leaves open the question of whether quantum teleportation is possible in the case of a finite group of reference frame transformations, without additional resources or prior communication. The protocol we exhibit provides an affirmative answer to this question.

A number of other solutions for the finite case have been proposed. In \cite{Chiribella2012}, it was suggested that Alice could transmit half of a maximally--entangled token state, in the regular representation, in advance of performing the protocol; the two parties could then use this to synchronise their operations. This method, however, requires Alice to be able to initialise an entangled state on a pair of systems each carrying the regular representation of the transformation group, a procedure which may be experimentally difficult or impossible.

Another relevant result can be found in \cite{Kitaev2004}, where it was shown that it is possible to perform quantum protocols using a second shared system as a quantum reference frame. This general result, intended for application to quantum cryptography, may be used to create reference frame--independent teleportation protocols. These protocols are formally identical to the token state method, but operationally more practicable; Alice and Bob simply initialise an additional shared entangled state at the same time as they create the first, take half each and use it to synchronise their operations. The problems of the token state method therefore persist in this case, although without the additional difficulty of communicating half of the second entangled state from one party to the other.

Finally, various solutions have also been proposed which use prior communication to align both parties' reference frames in advance of performing a normal teleportation protocol; see \cite{Bartlett2007}. This increases the amount of classical information that must be communicated for successful teleportation. Moreover, this procedure is not robust against changes in the alignment of Alice and Bob's reference frames, which must stay constant if the protocol is to succeed. Our protocol, in contrast, is robust against reference frame changes even while the classical message from Alice to Bob is in transit. We only require that the alignment of Bob's reference frame is constant for the short time between his receipt of the classical information and his application of the unitary correction.

However, we emphasize that our solution applies only when the system to be teleported carries a representation of the group of reference frame transformations $G$ which admits a $G$-equivariant unitary error basis. The other approaches described above do not have this limitation.

\section{Example of the procedure}
\label{sec:example}

In this section we give an informal account of the problem of reference frame--independent quantum teleportation, in the specific case where the systems are two-dimensional and the reference frame corresponds to a choice of spatial direction. This is followed by a more general and mathematically precise treatment in the next section.



Alice and Bob are quantum information theorists in separate laboratories, which do not necessarily have the same orientation in space. However, their relative orientations are not completely unknown: we are given some finite subgroup $G \subset SO(3)$, the group of rigid spatial rotations, with the promise that there is some element $g \in G$ which relates Alice's and Bob's frames. The group $G$ is common knowledge to both parties.

The task is to perform teleportation of a quantum state from Alice to Bob, without revealing their spatial orientations, either to each other or to a potential eavesdropper. There are a variety of reasons why this may be advantageous: this information may be strategically or cryptographically valuable, and hence they may prefer not to divulge it for reasons of privacy; they may prefer to conserve limited bandwidth, and hence  to not communicate redundant reference frame alignment information if it can be avoided; or they may simply be disoriented, and not aware of their own orientations.

In this example, we consider the case that $G=\mathbb{Z}_2$, meaning that their laboratories may be in one of two possible orientations which are related by a 180${}^\circ$ rotation about some given axis. We suppose that the nontrivial element $a \in G$ acts on the qubit to be teleported as follows:\footnote{It will be seen later that the specific choice of $\pi(a)$ is irrelevant, and all that determines whether reference frame--independent teleportation is possible is the isomorphism class of the representation of~$G$.}
$$
\pi(a) = \begin{pmatrix}
\sqrt{3}/2 & 1/2 \\ 1/2 & -\sqrt{3}/2
\end{pmatrix}.
$$
The fact that this matrix is nontrivial corresponds to the fact that the quantum system is not symmetric under the rotation operation.

Alice and Bob agree in advance to perform quantum teleportation  as follows. Alice will measure her initial system together with her part of the entangled state in the basis $\ket {\phi_i}$ and communicate the result to Bob, who will apply the corresponding correction $U_i$. We define $\ket {\phi_i} = (\mathbbm 1 \otimes U_i^\T) \ket \eta$, where $U_i ^\T$ denotes the transpose of the matrix $U_i$, the symbol $\mathbbm 1$ denotes the 2-by-2 identity matrix,  and $\ket\eta$ is the Bell state given above. The $U_i$ are defined as follows:
\begin{align*}
U_0 &=  \frac 1 {\sqrt{2}} \scriptsize \begin{pmatrix} 1 & 1 \\ -1 & 1 \end{pmatrix}
&
U_2 &= 
\frac 1 4 \scriptsize\begin{pmatrix} -\sqrt 2 - \sqrt 6 & - \sqrt 2 + \sqrt 6 \\ - \sqrt 2 + \sqrt 6 & \sqrt 2 + \sqrt 6 \end{pmatrix}
\\
U_1 &= \frac 1 {\sqrt{2}} \scriptsize \begin{pmatrix} 1 & -1 \\ 1 & 1 \end{pmatrix}
&
U_3 &= 
\frac 1 4 \scriptsize\begin{pmatrix} \sqrt 2 - \sqrt 6 & - \sqrt 2 - \sqrt 6 \\ - \sqrt 2 - \sqrt 6 & -\sqrt 2 + \sqrt 6 \end{pmatrix}
\end{align*}
It can readily be checked that this data forms an \emph{unitary error basis}, and so by the results of Werner~\cite{Werner2001} gives correct data for the execution of an ordinary quantum teleportation procedure for a single qubit, when the shared state is the Bell state $\ket \eta$.

If Bob's reference frame is correctly aligned with Alice's, then they are carrying out ordinary quantum teleportation, and the procedure will be successful. However, if his reference direction is upside-down with respect to Alice's, then teleportation does not proceed successfully. From Alice's perspective, Bob's correction is \textit{not} in fact the unitary $U_i$ corresponding to her measurement result, but rather the unitary $\pi(a)^\dag U_i \pi(a)$; a straightforward calculation then shows that Bob will receive a mixed state, and quantum information has been irrevocably lost. From Bob's perspective, he correctly applied the unitary $U_i$, but the teleportation failed because the measurement result $i$ that Alice communicated to him did not correspond to the state she actually measured, which was $( \mathbbm{1} \otimes \pi(a) ) \ket{\phi_i}$.


We now provide a resolution. In our new procedure, rather than communicating the two bits encoding the measurement result to Bob using their shared classical channel, Alice sends two \textit{physical objects} to Bob: arrows, of the sort a medieval archer might use. She orients these arrows according to the measurement result that she obtained, using the following encoding, where her reference direction is written as $\uparrow$:
\begin{align*}
0 &\mapsto \{ \uparrow \uparrow \}
&
1 &\mapsto \{ \downarrow \downarrow \}
&
2 &\mapsto \{ \uparrow \downarrow \}
&
3 &\mapsto \{ \downarrow \uparrow \}
\end{align*}
One at a time, she physically sends these arrows through space to Bob's laboratory.

Bob observes their local orientations and infers the measurement result 0, 1, 2 or 3 that Alice obtained. Suppose Bob's laboratory is correctly aligned with Alice's; then he will correctly infer Alice's measurement result, and he will apply the corresponding unitary correction. In this case, the two parties have executed a traditional quantum teleportation protocol, albeit one where the two classical bits of information were transferred from Alice to Bob in an unusual way.

Now we suppose that Bob's laboratory is aligned upside-down with respect to Alice's. If Alice attempts to send the message 0, 1, 2 or 3, Bob will receive it as 1, 0, 3 or 2 respectively, since the arrows will appear to him with the opposite orientations. Furthermore, just as before, when Bob applies a unitary $U_i$, its action is seen in Alice's reference frame as $\pi(a)^\dag U_i \pi(a)$. We now see the point of the entire construction: the unitary error basis is carefully chosen so that these effects cancel out. Indeed, the following equations can be easily verified:
\begin{align*}
\pi(a)^\dag U_1 \pi(a) &= U_0
&
\pi(a)^\dag U_3 \pi(a) &= U_2
\\
\pi(a)^\dag U_0 \pi(a) &= U_1
&
\pi(a)^\dag U_2 \pi(a) &= U_3
\end{align*}
As a result, the quantum teleporation will conclude successfully, even though Alice's and Bob's reference directions were misaligned.

In summary, by a careful choice of the unitary error basis, and by transferring the measurement result as unspeakable rather than speakable information, the quantum teleportation procedure can be carried out in a way which is robust against this restricted sort of reference frame error. Note in particular that only 2 bits of classical information were transferred from Alice to Bob, just as with the traditional teleportation procedure, and the Hilbert space of the entangled resource is of minimal dimension, so this procedure is \textit{tight} in the sense of Werner~\cite{Werner2001}. Also note that the unspeakable information Bob receives from Alice is uniformly random, since Alice's measurement results are; in particular, Bob receives no information during the protocol about the relative alignment between the two reference frames. Finally, it is clear that the procedure would succeed even if Bob's reference frame were constantly changing between the two alignments, as long as the alignment stays constant between Bob's receipt of the arrows and his application of the unitary correction.

\section{Mathematical description of the proposal}
\label{quantumteleportationsection}

\subsection{Traditional teleportation}

Teleportation is a well-understood procedure. It is traditionally formulated  under the assumption that both parties have aligned reference frames \cite{Bennett1993, Werner2001}. (Note that we only consider \textit{tight} quantum teleportation in this paper, where the state spaces of the initial system and of the entangled systems all have the same dimension, which is equal to the number of classical bits transferred, and the procedure succeeds with probability 1; this corresponds to the informal restriction of `no additional resources or prior communication'.) The traditional formulation is as follows.

\begin{procedure}[Teleportation without communication of unspeakable information]
\label{refteleportationopint}
Alice wants to teleport her state $\ket{\phi}$ to Bob; she has one half of a maximally--entangled bipartite state $\omega$, and Bob the other. She performs a measurement with respect to an orthonormal basis of effects $F_i$ on the bipartite system built from her initial system and her half of the entangled state. She sends the measurement result $x$ to Bob through a generic classical channel. Bob then performs a unitary operator $T_x$ on his half of the entangled state. The data $(\omega, \{F_i\},\{T_i\} )$ is \emph{correct} if Bob is guaranteed to receive the state $\ket \phi$ at the end of the procedure.
\end{procedure}

A complete description  of correct data $(\omega, \{F_i\}, \{T_i\})$ was given by Werner~\cite{Werner2001}, as follows.
\begin{definition}
For a Hilbert space $H$, a \textit{unitary error basis} is a basis of unitary operators $U_i \in B(H)$, which are orthonormal under the Hilbert-Schmidt inner product:
$$
\Tr(U_i ^\dag U_j ^{}) = \dim(H) \, \delta_{ij}
$$
\end{definition}
\begin{theorem}[Werner]
\label{Wernertighttheorem}
Up to equivalence, teleportation schemes for systems with Hilbert space $H$ are in one to one correspondence with unitary error bases on $H$.
\end{theorem}

\noindent
Under this correspondence, the shared entangled state $\omega$ is the Bell state $\sum_i \ket{i} \otimes \ket{i}$ for any orthonormal basis $\{\ket{0},\ket{1},\dots\}$. Alice measures in the maximally--entangled basis $\{\ket{\phi_0},\ket{\phi_1},\dots \}$, where $\ket{\phi_x} \in H \otimes H$ is defined as $\sum_i \ket{i} \otimes U_x \ket{i}$. Bob's correction for the measurement outcome $x$ is $U_x^T$.


\subsection{Reference frame--independent quantum teleportation}

We now fully describe the problem we solve.

\begin{problem}[Reference frame--independent quantum teleportation]\label{problemdescription}
Alice and Bob are spatially separated quantum information theorists capable of performing local operations and classical communication. Each party has one half of an entangled state which they created at some point in the past using a shared reference frame. Alice wants to use this entanglement to communicate a quantum state to Bob by teleportation. However, their reference frames are now misaligned; in some observer's fixed reference frame, their frame alignments will be described by the action of unknown elements $g_A,g_B$ of the group $G$ of reference frame transformations.

\begin{itemize}
\item[(i)] Is there a teleportation scheme---that is, a valid choice of measurement effects $F_x$ and corrections $T_x$---such that Procedure \ref{refteleportationopint} is guaranteed to teleport Alice's state to Bob regardless of the alignment of their reference frames?

\item[(ii)] If not, can we develop a different teleportation procedure, using only local operations and classical communication, that is guaranteed to teleport Alice's state to Bob regardless of the alignment of their reference frames?
\end{itemize} 

\end{problem}

We will now show that the answer to (i) is almost always negative. First we note the following lemma.
\begin{lemma}
\label{allmeasurementsmustbeintertwiners}
Under the conditions of Problem \ref{problemdescription}, Procedure \ref{refteleportationopint} will work for all reference frame alignments if and only if the operations $F_x$ and $T_x$ are intertwiners for the group action.
\end{lemma}

\begin{proof} 
We express the operations with reference to the original shared reference frame. Let Alice and Bob's frame shifts be described by group elements $g_A$ and $g_B$ respectively. Alice measures $F_x$ relative to her reference frame; in the original frame the operation she has performed is $\pi(g_A)^{\dagger} F_x \pi(g_A)$. She then sends the result $x$ to Bob, who performs the operation $T_x$ relative to his frame; in the original frame the operation he has performed will be $\pi(g_B)^{\dagger} T_x \pi(g_B)$. In general, the channel will therefore only work for all reference frame configurations when, for all $g_A, g_B$, $\pi(g_A)^{\dagger} F_x \pi(g_A) = F_{g_A(x)}$ and $\pi(g_B)^{\dagger} T_x \pi(g_B) = T_{g_A(x)}$, for some action of $G$ on the set of measurement outcomes. Since $\pi(e)^{\dagger} T_x \pi(e) = T_x$ for the identity $e$, this clearly implies that $g(x)=x$ for all $g$. The result follows.
\end{proof}

\noindent We now demonstrate that Procedure \ref{refteleportationopint} works only for a trivial $G$-action, rendering it inadequate for reference frame--independent teleportation in any nontrivial case.
\begin{proposition}\label{tightoldteleportationistrivial}
Procedure \ref{refteleportationopint} will only work for all reference frame alignments when $G$ acts by a global phase.
\end{proposition}

\begin{proof}
By Theorem \ref{Wernertighttheorem} and Proposition \ref{allmeasurementsmustbeintertwiners}, Procedure \ref{refteleportationopint} will work only if all projections $\ket{\phi_i}\bra{\phi_i}$ and corrections $U_i$ are intertwiners. By the definition of $\ket{\phi_i}$ in Theorem \ref{Wernertighttheorem}, it is sufficient that all $U_i$ be intertwiners. Let us assume that this is the case. Since the $G$-action is trivial on a basis of $\End(H)$, it must be completely trivial on $\End(H)$. Therefore we have $H \otimes H^* \simeq n \cdot \mathbbm{1}$. By straightforward character theory, there can only be one copy of $\mathbbm{1}$ in the product of an irreducible representation with its dual. Breaking $H$ up into simple factors, it follows by counting dimensions that they must all be identical and one dimensional.
\end{proof}

\subsection{Our new scheme}

In answer to (ii), we will now present our new scheme for teleportation using unspeakable information transfer.
\begin{procedure}[Teleportation with communication of unspeakable information]
\label{chargeteleportationopint}
Alice wants to teleport her state $\ket \phi$ to Bob; she has one half of a maximally-entangled bipartite state $\omega$, and Bob the other. She forms the bipartite system given by her initial system together with her half of the entangled state, and sends it to Bob through a classical channel which is decoherent in the basis $F_x$. Bob then performs a unitary operator $T_x$ on his half of the entangled state. The data $(\omega, F_x, T_x)$ is \emph{correct} if Bob is guaranteed to receive the state $\ket \phi$ at the end of the procedure.
\end{procedure}

\noindent

\begin{remark}
The key aspect of Procedure \ref{chargeteleportationopint} is that misalignment of reference frames will not affect the way Bob perceives the data arriving through a generic classical channel, but it will affect his perception of the decohered bipartite system. In other words, the information Bob receives from Alice will depend in a nontrivial way on the alignment of his reference frame.
\end{remark}

The basic data of Procedure~\ref{chargeteleportationopint} is the same as for Procedure~\ref{refteleportationopint}, and so Theorem~\ref{Wernertighttheorem} still applies. However, not all unitary error bases give rise to successful teleporation schemes under this procedure. We now investigate which of them do.

\begin{definition}[$G$-equivariant unitary error basis]
For a Hilbert space $H$ equipped with a unitary representation of $G$, a unitary error basis is \emph{$G$-equivariant} when the elements are permuted by the natural action $M \mapsto \pi(g) M \pi(g)^{\dagger}$ of $G$ on $\End(H)$. Explicitly, $\pi(g) U_i \pi(g)^{\dagger} = U_{\sigma_g(i)}$ for some permutation $\sigma_g$ of the set $\set{1,\dots,d^2}$.
\end{definition}

\begin{theorem}
Procedure \ref{chargeteleportationopint} will succeed for any reference frame misalignment $g \in G$ just when the unitary error basis $U_i$ is $G$-equivariant.
\end{theorem}

\begin{proof}
We again work in Alice and Bob's original lab frame. Alice decoheres in the orthonormal basis $\{ \pi(g_A) \ket{\phi_0},  \pi(g_A) \ket{\phi_1}, \dots\}$. Bob then measures in the orthonormal basis $\{ \pi(g_B) \ket{\phi_0},  \pi(g_B) \ket{\phi_1}, \dots\}$, and, depending on his measurement outcome $x$, performs the corresponding correction $\pi(g_B) U^{T}_{x} \pi(g_B) ^{\dagger}$.

We first note that, putting Alice's decoherence and Bob's measurement together as one operation, we get a teleportation scheme under Definition~\ref{refteleportationopint}. Therefore, by Theorem \ref{Wernertighttheorem}, Alice's decoherence operation followed by Bob's measurement must be a measurement in some orthonormal basis of maximally--entangled states; clearly that must be the basis that Bob measures in. Letting Bob's measurement channel be $M_1$ and Alice's decohering channel be $M_2$, it follows that $M_1 \circ M_2 = M_1$; this can clearly only be true if the projection basis for $M_2$ is the same as the projection basis for $M_1$. We therefore see that the basis $\{ \pi(g_A) \ket{\phi_0},  \pi(g_A) \ket{\phi_1}, \dots\}$ must be some reordering of the basis $\{ \pi(g_B) \ket{\phi_0},  \pi(g_B) \ket{\phi_1}, \dots\}$, for all $g_A, g_B$. This is exactly $G$-equivariance of the UEB $U_i$.

We now demonstrate that this condition is sufficient to guarantee success for Procedure \ref{chargeteleportationopint}. Suppose $U_i$ is $G$-equivariant. Then Alice's decohering operation is exactly the same as it would have been if her reference frame had not shifted at all. Bob measures \emph{and} performs the correction; the correction therefore corresponds to the measurement and the result follows.
\end{proof}

In Procedure \ref{chargeteleportationopint}, we have specified that Alice send the decohered bipartite system itself, since this is always theoretically possible. However, in practise it may be experimentally more practicable to use some other means of classically communicating unspeakable information; the important thing is that the classical data should \emph{itself} carry the same $G$-action as the corresponding $G$-equivariant measurement basis. An example was given in Section~\ref{sec:example}, where Alice's measurement result was encoded in the spatial orientations of some physical objects. In order to demonstrate that this approach is applicable to other types of reference frame uncertainty, we provide two further examples of unspeakable encodings of classical information.

\begin{example}
\begin{enumerate}
\item[(i)] \emph{Time.}\, Suppose the computational basis states of Alice and Bob's systems are nondegenerate energy eigenstates (for example, eigenstates of the photon number operator). Here they will need to share a time reference. Let the time translation operator $U(t)$ have periodicity $U(t+T) = U(t)$ for some $T$. Suppose that the group $U(1)$ of time translations has been discretised to some cyclic subgroup $\mathbb{Z}_n$ of translations by $T/n$. Alice and Bob's reference frame configurations will correspond to their zeroes of time. Signals sent by Alice to to Bob which arrive, according to her reference frame, at time $m_AT/n$, arrive for Bob at a different time $m_BT/n$, depending on the difference between their reference frames. By encoding her measurement result in the time of arrival of signals, Alice may construct $G$-equivariant teleportation protocols.

\item[(ii)] \emph{Circular polarisation.}\, Suppose Alice and Bob are working with photonic qubits whose computational basis states are left and right circular polarisation. In this case, the group of reference frame transformations will correspond to planar rotations of the axes perpendicular to the propagation direction. Suppose that the group $U(1)$ of planar rotations has been discretised to some cyclic subgroup $\mathbb{Z}_n$ of rigid rotations by multiples of $2\pi/n$, and that Alice can classically communicate linearly polarised light to Bob. By communicating frame configurations using beams of linearly polarised light, Alice may encode measurement results in the angle difference between Bob's frame and the commmunicated frame, allowing her to construct $G$-equivariant teleportation protocols.
\end{enumerate}
\end{example}

\ignore{\begin{remark}
\JVcomm{This is interesting but not essential I\ think; we could drop it if we need the space.}This procedure, in addition to teleporting Alice's state, also allows Bob to deduce information about the relative alignment of his and Alice's reference frames. In particular, if Alice is able to measure the system without destroying it, and classically communicates to Bob her readout in addition to the decohered state, Bob, if he knows the permutation action of $G$ on the unitary error basis, can compare that result to his own; the element of the group describing the transformation from Alice's frame to his own must permute those two measurement results. 
\end{remark}
}

\section{Classical structures in $\Rep(G)$}
\label{sec:repg}

Teleportation in the context of a finite group $G$ can be described elegantly in the framework of \textit{categorical quantum mechanics}~\cite{Abramsky2009}. One key strategy in this research programme is to understand features of quantum information in terms of the category \cat{FHilb} of finite-dimensional Hilbert spaces and linear maps, and then to generalize them by applying them in different categories. The concept of $G$-equivariant quantum teleportation arises by understanding the categorical structure of the traditional quantum teleportation procedure, and then applying it in $\cat{Rep}(G)$, as we now explore. This technical section of the paper will make use of well-known ideas from categorical quantum mechanics, of which full details are available in the provided references.

The following definition gives our abstract categorical description of quantum teleportation, in terms of classical structures in a symmetric monoidal category~\cite{Coecke2013}.
\begin{definition}
In a dagger-compact category, a \textit{quantum teleportation procedure} on an object $A$ with a right dual is a classical structure on the object $A \otimes A^*$, satisfying the following condition, where $c$ is some scalar:
\tikzset{cqm/.style={line width=0.7pt}}
\begin{equation}
\begin{aligned}
\begin{tikzpicture}[cqm]
\draw (0,-0.5) to [out=up, in=down] (0,1);
\draw (1,1) to (1,0.5) to [out=down, in=down, looseness=1] (3,0.5) to (3,2.5) to [out=up, in=up, looseness=1.5] (2,2.5) to (2,2);
\draw [fill=blue!10] (-1.25,1) rectangle +(3.5,1);
\draw (1,2) to +(0,1.5);
\draw (0,2) to +(0,1.5);
\draw (-1,2) to +(0,1.5);
\node at (0.5,1.5) {comultiplication};
\draw [->] (0,0) to (0,0.25);
\draw [->] (1,1) to (1,0.5);
\draw [->] (3,1) to (3,1.5);
\draw [->] (3,1) to (3,1.5);
\draw [->] (-1,2) to +(0,0.75);
\draw [-<] (0,2) to +(0,0.75);
\draw [->] (1,2) to +(0,0.75);
\draw [-<] (2,2) to +(0,0.5);
\end{tikzpicture}
\end{aligned}
\quad=\quad
c\,\,\,\,\cdot\,\,\,\,
\begin{aligned}
\begin{tikzpicture}[cqm]
\draw (1,-0.5) to [out=up, in=down] (1,3.5);
\draw [fill=blue!10] (-1.25,1) rectangle +(1.5,1);
\draw (0,2) to +(0,1.5);
\draw (-1,2) to +(0,1.5);
\node at (-0.5,1.5) {unit};
\draw [->] (-1,2) to +(0,0.75);
\draw [-<] (0,2) to +(0,0.75);
\draw [->] (1,0) to +(0,1.5);
\end{tikzpicture}
\end{aligned}
\end{equation}
\end{definition}

\noindent
This definition is motivated by the following theorem; recall Werner's Theorem~\ref{Wernertighttheorem}.
\begin{theorem}
Quantum teleportation procedures in \cat{FHilb} correspond precisely to unitary error bases.
\end{theorem}

We now summarize the application of these ideas in a group representation category.
\begin{definition}
For a group $G$, the dagger-compact category $\cat{Rep}(G)$ has objects given by unitary representations of $G$, morphisms given by intertwiners, and a dagger-compact structure inherited from the underlying Hilbert spaces.
\end{definition}
\begin{theorem}
Quantum teleportation procedures in $\cat{Rep}(G)$ correspond precisely to $G$-equivariant unitary error bases.
\end{theorem}

\noindent
Finally, we observe that the constructions of unitary error bases in Theorem \ref{Hadamardconstruction} and Remark \ref{QLSremark} carry over straightforwardly to the $G$-equivariant setting because they are essentially categorical constructions; the Hadamard construction, for instance, is defined in terms of two special commutative dagger Frobenius algebras and an isomorphism between them. In {\bf Rep($G$)}, this reduces exactly to the intertwining Hadamard matrix and $G$-equivariant orthonormal basis of Theorem \ref{Hadamardconstruction}. In this sense, these constructions are much more natural than, for instance, the construction of unitary error bases using projective group representations \cite{Knill1996}; indeed, it is difficult to see how the latter construction could be brought into the $G$-equivariant framework.

\section{Existence and construction of RFI teleportation protocols}
\label{existenceconstruction}

We have demonstrated that $G$-equivariant UEBs are exactly the structures we need to perform reference frame--independent teleportation protocols, but it is still unclear how to construct them for a given representation $H$, if they exist at all. We cannot hope for a general classification of $G$-equivariant UEBs, since there is not even a classification in the case where the $G$-action is trivial\footnote{The problem of classifying UEBs is closely related to the difficult problem of classifying Hadamard matrices~\cite{SzollHosi2011}.}, although many construction methods exist~\cite{Werner2001, Musto2015, Knill1996}. In this section we will demonstrate that $G$-equivariant unitary error bases need not exist on every representation, meaning that RFI teleportation is not always possible.  We will then demonstrate that several UEB constructions carry over naturally to the $G$-equivariant setting, allowing us to construct RFI teleportation protcols for a wide variety of systems.

We begin with a definition.

\begin{definition}
A \emph{$G$-equivariant orthonormal basis} for some representation $V$ is an orthonormal basis of $V$ whose elements are permuted by the action of $G$.
\end{definition}

\begin{remark}\label{gequebsaregeqobs}
$G$-equivariant unitary error bases are $G$-equivariant orthonormal bases of $\End(H) \simeq H \otimes H^*$, all of whose elements are unitary maps.
\end{remark}

\noindent
It will transpire that we can use $G$-equivariant orthonormal bases on $H$ to construct $G$-equivariant UEBs for $H$. Moreover, if we prove that there are no $G$-equivariant orthonormal bases on $\End(H)$, it follows by Remark \ref{gequebsaregeqobs} that there will be no $G$-equivariant UEBs for $H$; we will use this fact to demonstrate that RFI teleportation protocols need not always exist. Our first step is therefore a classification of $G$-equivariant orthonormal bases.

\subsection{A classification of $G$-equivariant orthonormal bases}\label{geqobclassification}

We begin with a simple lemma. Let {\bf $G$-Set} be the category whose objects are sets carrying an action of $G$, and whose morphisms are $G$-equivariant functions between them. Then there exists a functor $\mathcal{M}: \textbf{$G$-Set} \to  \textbf{Rep($G$)}$, which, given a $G$-set, constructs the free Hilbert space on its elements, and extends the $G$-action and morphisms linearly.

\begin{lemma}
$G$-equivariant orthonormal bases exist only on representations isomorphic to those in the image of $\mathcal{M}$.
\end{lemma}

\begin{proof}
Immediate, since a $G$-equivariant orthonormal basis has an underlying Hilbert space isomorphic to the free Hilbert space on the elements of the chosen basis, which $G$ acts on by permutations.
\end{proof}

We begin by presenting a simple classification of $G$-sets due to Burnside~\cite{Burnside2004}.

\begin{definition}

Given two $G$-sets $(X_1, \sigma_1)$ and $(X_2,\sigma_2)$, their \emph{disjoint union} $(X_1 \sqcup X_2, \sigma_1 \sqcup \sigma_2)$ is the disjoint union of $X_1$ and $X_2$ as sets with the natural induced action.
\end{definition}

\begin{definition}

 Given a subgroup $H$ of $G$, the \emph{coset space} $(G/H, \sigma_H)$ is the $G$-set whose elements are the cosets of $H$ in $G$, and whose $G$-action $\sigma_H$ is the natural action of $G$ by left multiplication on those cosets. 

\end{definition}
\begin{lemma}\label{Gsetdecomp}
Any $G$-set is isomorphic to a disjoint union of coset spaces. Two coset spaces are isomorphic as $G$-sets if and only they correspond to conjugate subgroups.
\end{lemma}

\begin{proof}
See \cite{Burnside2004}.
\end{proof}

\begin{remark}
In modern language, Lemma \ref{Gsetdecomp} states that {\bf $G$-Set} is a semisimple fusion category whose simple objects correspond to conjugacy classes of subgroups in $G$. It is easy to see that the functor $\mathcal{M}$ is additive; the disjoint union of two $G$-sets will be sent under $\mathcal{M}$ to the direct sum of their corresponding representations. In order to classify all objects in the image of $\mathcal{M}$, therefore, it is sufficient to find the image of the coset spaces under $\mathcal{M}$. We will call those representations the \emph{basic permutation representations}.
\end{remark}

In order to identify the basic permutation representations, we now state an obvious but critical lemma regarding the character of the permutation representation induced by $\mathcal{M}$ on a $G$-set.

\begin{lemma}
Given a $G$-set $(X,\sigma)$, let $\chi: G \to \mathbb{R}$ be the character of $\mathcal{M}(X,\sigma)$. Then the following holds:
\begin{equation}
\chi(g) = |\Set{x \in X | g \cdot x = x } |
\end{equation}
\end{lemma}

\begin{proof} The character $\chi(g)$ is exactly the trace of the matrix representing $g$; the result follows trivially from the definition of $\mathcal{M}(X,\sigma)$.
\end{proof}

We may therefore identify the basic permutation representations by taking a representative of every conjugacy class of subgroups of $G$, finding the number of fixed points of the action of each element of $G$ on the corresponding coset spaces, then decomposing the resulting characters using the character table to find the corresponding representations.

\subsection{Existence of RFI teleportation protocols}

Using the results of Subsection \ref{geqobclassification}, we now exhibit a representation for which no $G$-equivariant UEBs exist, and on which quantum teleportation is therefore impossible.

\begin{proposition}
There is no RFI protocol to teleport the state of the 2-dimensional irreducible representation $V$ of $S_3$.
\end{proposition}

\begin{proof}
Using the method outlined in Subsection \ref{geqobclassification}, we find that the characters of the basic permutation representations are as follows:
\begin{center}
\begin{tabular}{|l | | l | l | l|}
\hline
Representation / Conjugacy class & () & (1,2) & (1,2,3) \\
\hline
$\mathcal{M}(G/C_1)$ & 6 & 0 & 0     \\
\hline
$\mathcal{M}(G/C_2)$ & 3 & 1 & 0     \\
\hline 
$\mathcal{M}(G/C_3)$ & 2 & 0 & 2     \\
\hline
$\mathcal{M}(G/C_4)$ & 1  & 1  & 1   \\
\hline
\end{tabular}
\end{center}
The character of $V \otimes V^*$ is $4|0|1$, which clearly cannot be composed as a sum of characters of basic permutation representations. By Remark \ref{gequebsaregeqobs}, the result follows. 
\end{proof}

\begin{remark} 
This argument does not extend to all irreducible representations. The endomorphism space of the 2-dimensional irreducible representation of $D_8$, for instance, is a sum of basic permutation representations.
\end{remark}

\subsection{Construction of RFI teleportation protocols}

Although RFI teleportation protocols need not always exist, they can often be constructed. We now demonstrate that, if we can find a $G$-equivariant orthonormal basis on $H$, and a Hadamard matrix which commutes with all $\pi(G)$ in that basis, we can perform RFI teleportation on $H$.

\begin{theorem} \label{Hadamardconstruction}
Let $\ket{v_i}$ be a $G$-equivariant orthonormal basis on $H$. In this basis all $\pi(g)$ will be permutation matrices. Let $H$ be a Hadamard matrix that commutes with all $\pi(g)$ in this basis. Then the following family is a $G$-equivariant UEB: 
\begin{equation}
(U_H)_{ij} = \frac{1}{N}H \circ \diag(H,j)^{\dagger} \circ H^{\dagger} \circ \diag(H^T,i)
\end{equation}
\end{theorem}

\begin{proof}
It was already proved in \cite{Musto2015} that this is a UEB; we therefore need only show that it is $G$-equivariant. Since $H \in C_{U(n)}(G)$ we have that $g U_{ij} g^{\dagger} =  \frac{1}{n} H \circ (g \circ \diag(H,j)^{\dagger} \circ g^{\dagger}) \circ H^{\dagger} \circ (g \circ \diag(H^T,i) \circ g^{\dagger}).$
We see easily that $g \circ \diag(H,j)^{\dagger} \circ g^{\dagger} = g \circ \diag(H^{*},j) \circ g^{\dagger} = \diag(H^{*} \circ g, j)$. Now note that the fact that $H$ commutes with all elements of $G$ means that permuting the columns of $H$ is exactly the same as permuting the rows, since $gH=Hg$ for all $g \in G$. So $\diag(H^{*} \circ g, j) = \diag(g \circ H^{*}, j) = \diag(H^{*}, g \cdot j)$. A similar argument works for $\diag(H^T,i)$.
\end{proof}

\begin{remark}\label{QLSremark}
If the assumptions of Theorem \ref{Hadamardconstruction} are satisfied, it is possible to construct many more $G$-equivariant UEBs using \emph{quantum Latin squares} (QLSs) \cite{Musto2015}; this construction will give $G$-equivariant UEBs provided the linear map defining the QLS is an intertwiner.
\end{remark}

We finish this section with a simple sufficient condition for the existence of tight RFI protocols on systems of dimension less than 5. Firstly we prove a lemma.

\begin{lemma}\label{unitarymatrixconditions}
Let $M$ be a matrix of dimension $\geq 3$ defined by two complex parameters $a$ and $b$, where all entries on the diagonal are $a$, and all other entries are $b$. Let $a=|a|\alpha, b= |b|\beta $ where $\alpha,\beta \in U(1)$ and $|a|,|b| \neq 0$. Then $M$ is unitary precisely when the following conditions are satisfied:

\noindent\begin{minipage}{0.3\linewidth}
\begin{equation}
\label{inequalityfora} \frac{n-2}{n} \leq |a| \leq 1
\end{equation}
\end{minipage}%
\begin{minipage}{0.33\linewidth}
\begin{equation}
 \label{equalityforb} |b|^2 = \frac{1 - |a|^2}{n-1}
\end{equation}
\end{minipage}%
\begin{minipage}{0.36\linewidth}
\begin{equation}
 \Re(\alpha^* \beta) = \frac{2-n}{2} \frac{|b|}{|a|}
\end{equation}
\end{minipage}\par\vspace{0pt}

\end{lemma}

\begin{proof}
For unitarity it is sufficient that the rows form an orthonormal basis. It is clear from the symmetry of $Q$ that it is sufficient for one row vector to be normal, and one pair of row vectors to be orthogonal. This gives us two equations in $a$ and $b$:
\begin{align}
|b|^2 &= \frac{1 - |a|^2}{n-1}  \label{unq1} \\
\Re(a^*b) &= \frac{2-n}{2} |b|^2. \label{unq2}   
\end{align}
We will demonstrate that (\ref{inequalityfora}) is necessary and sufficient  for us to find $b$ satisfying these equations. It is obvious that (\ref{unq1}) is satisfiable if and only if $|a| \leq 1$.
Letting $a=|a|\alpha, b= |b|\beta$, Equation (\ref{unq2}) becomes 
$$
\Re(\alpha^* \beta) = \frac{2-n}{2} \frac{|b|}{|a|}.
$$
Since $-1 \leq \Re(\alpha^* \beta) \leq 1$ and $\alpha, \beta$ can be freely adjusted to give $ Re(\alpha^* \beta)$ any value in that range, we see that the following is necessary and sufficient for (\ref{unq2}) to be soluble:
$$ \frac{(2-n)^2}{4} \frac{|b|^2}{|a|^2} \leq  1 $$
Use of the identity ($\ref{unq1}$) and a short calculation demonstrates that this is equivalent to the lower bound in the inequality (\ref{inequalityfora}).
\end{proof}

\begin{theorem}
Suppose $H$ admits a $G$-equivariant orthonormal basis, and is of dimension less than 5. Then there exists a RFI teleportation protocol for $H$.
\end{theorem}

\begin{proof}
We construct a $G$-equivariant UEB for $H$.
Expressed in the $G$-equivariant orthonormal basis, $\pi(G)$ will be some subgroup of the permutation matrices $S_n$. To use Theorem \ref{Hadamardconstruction}, we must find a Hadamard matrix commuting with $\pi(G)$. In the worst case, $\pi(G)$ will be the whole group $S_n$ of permutation matrices. (This situation is realised for the representation $1 \oplus V$ of $\mathfrak{S}_n$, where $V$ is the fundamental $(n-1)$-dimensional representation of $\mathfrak{S}_n$).

We will demonstrate that, when $H$ is of dimension less than 5, we can find a Hadamard matrix which commutes with all the permutation matrices. First we eliminate the degenerate cases $n=1$ and $n=2$. Clearly for $n=1$ we can perform RFI teleportation by Proposition \ref{tightoldteleportationistrivial}, For $n=2$ the following family of Hadamard matrices commutes with $S_2$, where $|a|=|b|=1/\sqrt{2}$ and $Re (a^*b)=0$:
$$
\begin{pmatrix}
a & b \\
b & a
\end{pmatrix}
$$
From now on we may therefore assume $n \geq 3$.

It is easy to see that the centraliser $C_{M_n}(S_n) \subset M_n$ is the set of matrices defined by two complex parameters $a$ and $b$, where all entries on the diagonal are $a$, and all other entries are $b$. 
The conditions necessary for such a matrix to be unitary were given in Lemma \ref{unitarymatrixconditions}.
Setting $|a|=|b|$ in (\ref{equalityforb}), it follows that
$|a| = 1/\sqrt{n}.$ This is compatible with (\ref{inequalityfora}) only for $n \leq 4$.
\end{proof}

\bibliographystyle{eptcs}
\bibliography{FiniteGroupTeleportation}

\begin{thebibliography}{10}
\providecommand{\bibitemdeclare}[2]{}
\providecommand{\surnamestart}{}
\providecommand{\surnameend}{}
\providecommand{\urlprefix}{Available at }
\providecommand{\url}[1]{\texttt{#1}}
\providecommand{\href}[2]{\texttt{#2}}
\providecommand{\urlalt}[2]{\href{#1}{#2}}
\providecommand{\doi}[1]{doi:\urlalt{http://dx.doi.org/#1}{#1}}
\providecommand{\bibinfo}[2]{#2}

\bibitemdeclare{incollection}{Abramsky2009}
\bibitem{Abramsky2009}
\bibinfo{author}{Samson \surnamestart Abramsky\surnameend} \&
  \bibinfo{author}{Bob \surnamestart Coecke\surnameend} (\bibinfo{year}{2009}):
  \emph{\bibinfo{title}{Categorical Quantum Mechanics}}.
\newblock In \bibinfo{editor}{Dov M.~Gabbay \surnamestart
  Daniel~Lehmann\surnameend, Kurt~Engesser}, editor: {\sl
  \bibinfo{booktitle}{Handbook of Quantum Logic and Quantum Structures}},
  \bibinfo{publisher}{Elsevier}, pp. \bibinfo{pages}{261--323},
  \doi{10.1016/B978-0-444-52869-8.50010-4}.

\bibitemdeclare{article}{Bartlett2007}
\bibitem{Bartlett2007}
\bibinfo{author}{Stephen \surnamestart Bartlett\surnameend},
  \bibinfo{author}{Terry \surnamestart Rudolph\surnameend} \&
  \bibinfo{author}{Rob \surnamestart Spekkens\surnameend}
  (\bibinfo{year}{2007}): \emph{\bibinfo{title}{Reference frames,
  superselection rules, and quantum information}}.
\newblock {\sl \bibinfo{journal}{Rev. Mod. Phys.}} \bibinfo{volume}{79}, pp.
  \bibinfo{pages}{555--609}, \doi{10.1103/RevModPhys.79.555}.

\bibitemdeclare{article}{Bennett1993}
\bibitem{Bennett1993}
\bibinfo{author}{Charles \surnamestart Bennett\surnameend},
  \bibinfo{author}{Giles \surnamestart Brassard\surnameend},
  \bibinfo{author}{Claude \surnamestart Cr\'epeau\surnameend},
  \bibinfo{author}{Richard \surnamestart Jozsa\surnameend},
  \bibinfo{author}{Asher \surnamestart Peres\surnameend} \&
  \bibinfo{author}{William~K. \surnamestart Wootters\surnameend}
  (\bibinfo{year}{1993}): \emph{\bibinfo{title}{Teleporting an unknown quantum
  state via dual classical and Einstein-Podolsky-Rosen channels}}.
\newblock {\sl \bibinfo{journal}{Phys. Rev. Lett.}} \bibinfo{volume}{70}, pp.
  \bibinfo{pages}{1895--1899}, \doi{10.1103/PhysRevLett.70.1895}.

\bibitemdeclare{book}{Burnside2004}
\bibitem{Burnside2004}
\bibinfo{author}{William \surnamestart Burnside\surnameend}
  (\bibinfo{year}{2004}): \emph{\bibinfo{title}{Theory of Groups of Finite
  Order}}.
\newblock \bibinfo{series}{Dover Books on Mathematics Series},
  \bibinfo{publisher}{Dover Publications}.

\bibitemdeclare{article}{Chiribella2012}
\bibitem{Chiribella2012}
\bibinfo{author}{Giulio \surnamestart Chiribella\surnameend},
  \bibinfo{author}{Vittorio \surnamestart Giovannetti\surnameend},
  \bibinfo{author}{Lorenzo \surnamestart Maccone\surnameend} \&
  \bibinfo{author}{Paolo \surnamestart Perinotti\surnameend}
  (\bibinfo{year}{2012}): \emph{\bibinfo{title}{Teleportation transfers only
  speakable quantum information}}.
\newblock {\sl \bibinfo{journal}{Phys. Rev. A}} \bibinfo{volume}{86}, p.
  \bibinfo{pages}{010304}, \doi{10.1103/PhysRevA.86.010304}.

\bibitemdeclare{article}{Coecke2013}
\bibitem{Coecke2013}
\bibinfo{author}{Bob \surnamestart Coecke\surnameend}, \bibinfo{author}{Dusko
  \surnamestart Pavlovic\surnameend} \& \bibinfo{author}{Jamie \surnamestart
  Vicary\surnameend} (\bibinfo{year}{2013}): \emph{\bibinfo{title}{A new
  description of orthogonal bases}}.
\newblock {\sl \bibinfo{journal}{Mathematical Structures in Computer Science}}
  \bibinfo{volume}{23}, pp. \bibinfo{pages}{555--567},
  \doi{10.1017/S0960129512000047}.

\bibitemdeclare{article}{Enk2001}
\bibitem{Enk2001}
\bibinfo{author}{Steven~J. \surnamestart van Enk\surnameend}
  (\bibinfo{year}{2001}): \emph{\bibinfo{title}{The physical meaning of phase
  and its importance for quantum teleportation}}.
\newblock {\sl \bibinfo{journal}{Journal of Modern Optics}}
  \bibinfo{volume}{48}(\bibinfo{number}{13}), pp. \bibinfo{pages}{2049--2054},
  \doi{10.1080/09500340108240906}.

\bibitemdeclare{phdthesis}{SzollHosi2011}
\bibitem{SzollHosi2011}
\bibinfo{author}{\surnamestart {Ferenc Sz{\"o}ll{\H o}si}\surnameend}
  (\bibinfo{year}{2011}): \emph{\bibinfo{title}{{Construction, classification
  and parametrization of complex Hadamard matrices}}}.
\newblock Ph.D. thesis, \bibinfo{school}{The University of Wisconsin, Madison}.

\bibitemdeclare{article}{Gisin2002}
\bibitem{Gisin2002}
\bibinfo{author}{Nicolas \surnamestart Gisin\surnameend},
  \bibinfo{author}{Gregoire \surnamestart Ribordy\surnameend},
  \bibinfo{author}{Wolfgang \surnamestart Tittel\surnameend} \&
  \bibinfo{author}{Hugo \surnamestart Zbinden\surnameend}
  (\bibinfo{year}{2002}): \emph{\bibinfo{title}{Quantum cryptography}}.
\newblock {\sl \bibinfo{journal}{Rev. Mod. Phys.}} \bibinfo{volume}{74}, pp.
  \bibinfo{pages}{145--195}, \doi{10.1103/RevModPhys.74.145}.

\bibitemdeclare{article}{Gottesman1999}
\bibitem{Gottesman1999}
\bibinfo{author}{Daniel \surnamestart Gottesman\surnameend} \&
  \bibinfo{author}{Isaac~L. \surnamestart Chuang\surnameend}
  (\bibinfo{year}{1999}): \emph{\bibinfo{title}{Demonstrating the viability of
  universal quantum computation using teleportation and single-qubit
  operations}}.
\newblock {\sl \bibinfo{journal}{Nature}}
  \bibinfo{volume}{402}(\bibinfo{number}{6760}), pp. \bibinfo{pages}{390--393},
  \doi{10.1038/46503}.

\bibitemdeclare{article}{Kitaev2004}
\bibitem{Kitaev2004}
\bibinfo{author}{Alexei \surnamestart Kitaev\surnameend},
  \bibinfo{author}{Dominic \surnamestart Mayers\surnameend} \&
  \bibinfo{author}{John \surnamestart Preskill\surnameend}
  (\bibinfo{year}{2004}): \emph{\bibinfo{title}{Superselection rules and
  quantum protocols}}.
\newblock {\sl \bibinfo{journal}{Phys. Rev. A}} \bibinfo{volume}{69}, p.
  \bibinfo{pages}{052326}, \doi{10.1103/PhysRevA.69.052326}.

\bibitemdeclare{book}{Knill1996}
\bibitem{Knill1996}
\bibinfo{author}{Emanuel \surnamestart Knill\surnameend}
  (\bibinfo{year}{1996}): \emph{\bibinfo{title}{Non-binary unitary error bases
  and quantum codes}}.
\newblock \doi{10.2172/373768}.

\bibitemdeclare{article}{Marvian2013}
\bibitem{Marvian2013}
\bibinfo{author}{Iman \surnamestart Marvian\surnameend} \&
  \bibinfo{author}{Robert~W. \surnamestart Spekkens\surnameend}
  (\bibinfo{year}{2013}): \emph{\bibinfo{title}{The theory of manipulations of
  pure state asymmetry: I. Basic tools, equivalence classes and single copy
  transformations}}.
\newblock {\sl \bibinfo{journal}{New Journal of Physics}}
  \bibinfo{volume}{15}(\bibinfo{number}{3}), p. \bibinfo{pages}{033001},
  \doi{10.1103/physreva.90.014102}.

\bibitemdeclare{article}{Marvian2014}
\bibitem{Marvian2014}
\bibinfo{author}{Iman \surnamestart Marvian\surnameend} \&
  \bibinfo{author}{Robert~W. \surnamestart Spekkens\surnameend}
  (\bibinfo{year}{2014}): \emph{\bibinfo{title}{Asymmetry properties of pure
  quantum states}}.
\newblock {\sl \bibinfo{journal}{Phys. Rev. A}} \bibinfo{volume}{90}, p.
  \bibinfo{pages}{014102}, \doi{10.1103/PhysRevA.90.014102}.

\bibitemdeclare{article}{Marzolino2015}
\bibitem{Marzolino2015}
\bibinfo{author}{Ugo \surnamestart Marzolino\surnameend} \&
  \bibinfo{author}{Andreas \surnamestart Buchleitner\surnameend}
  (\bibinfo{year}{2015}): \emph{\bibinfo{title}{Quantum teleportation with
  identical particles}}.
\newblock {\sl \bibinfo{journal}{Phys. Rev. A}} \bibinfo{volume}{91}, p.
  \bibinfo{pages}{032316}, \doi{10.1103/PhysRevA.91.032316}.

\bibitemdeclare{article}{Musto2015}
\bibitem{Musto2015}
\bibinfo{author}{Benjamin \surnamestart Musto\surnameend} \&
  \bibinfo{author}{Jamie \surnamestart Vicary\surnameend}
  (\bibinfo{year}{2016}): \emph{\bibinfo{title}{{Quantum Latin squares and
  unitary error bases}}}.
\newblock {\sl \bibinfo{journal}{Quantum Information and Computation}}.
\newblock \bibinfo{note}{To appear}.

\bibitemdeclare{article}{Peres2002}
\bibitem{Peres2002}
\bibinfo{author}{Asher \surnamestart Peres\surnameend} \&
  \bibinfo{author}{Petra~F. \surnamestart Scudo\surnameend}
  (\bibinfo{year}{2002}): \emph{\bibinfo{title}{{Unspeakable quantum
  information}}}.
\newblock
  \bibinfo{note}{\href{http://arxiv.org/abs/quant-ph/0201017}{quant-ph/0201017%
}}.

\bibitemdeclare{article}{Werner2001}
\bibitem{Werner2001}
\bibinfo{author}{Reinhard \surnamestart Werner\surnameend}
  (\bibinfo{year}{2001}): \emph{\bibinfo{title}{All teleportation and dense
  coding schemes}}.
\newblock {\sl \bibinfo{journal}{Journal of Physics A: Mathematical and
  General}} \bibinfo{volume}{34}(\bibinfo{number}{35}), p.
  \bibinfo{pages}{7081}, \doi{10.1088/0305-4470/34/35/332}.

\end{thebibliography}

\end{document}